\documentclass[10pt]{article}

\usepackage[T1]{fontenc}
\usepackage[tt=false, type1=true]{libertine}
\usepackage[varqu]{zi4}
\usepackage[libertine]{newtxmath}

\usepackage{booktabs} 
\usepackage{microtype}

\usepackage[a4paper, total={6.5in, 10.5in},top=0.8in, left=1in, includefoot]{geometry}

\usepackage[bottom]{footmisc}

\usepackage[svgnames]{xcolor}
	\definecolor{darkred}{rgb}{.6,0,0}
	\definecolor{darkgreen}{rgb}{0,0.6,0}
	\definecolor{darkblue}{rgb}{0,0.,0.6}

\usepackage{amsthm}
\usepackage{amssymb}
\usepackage{amsmath}
\usepackage{amsfonts}
\usepackage{diagbox} 
\usepackage{amscd}
\usepackage{bm}
\usepackage{dsfont}
\usepackage{enumitem}

\usepackage{graphicx} 
\usepackage{comment}
\usepackage{subfigure}
\usepackage[classfont=roman, langfont=roman, funcfont=italic]{complexity}

\usepackage{algorithm}
\usepackage{algorithmic}

\usepackage{multirow} 

\usepackage[backref=page, bookmarksnumbered]{hyperref}
	\hypersetup{
		colorlinks=true,
		allcolors = Black
	}

\usepackage{tikz}
    \usetikzlibrary{matrix,arrows, calc,decorations.pathmorphing, cd, shadows}
    \usetikzlibrary{shapes}
    \usetikzlibrary{fit}

\makeatletter
\g@addto@macro\@openbib@code{\setlength{\itemsep}{-0.1cm}}
\makeatother

\theoremstyle{plain}
\newtheorem{thm}{Theorem} 
\newtheorem{lem}{Lemma}
\newtheorem{prop}{Proposition} 
\newtheorem{cor}{Corollary}

\theoremstyle{definition}

\newtheorem{rem}{Remark}

\newcommand{\Abs}[1]{\left\lvert#1\right\rvert}
\newcommand{\abs}[1]{\lvert#1\rvert}
\newcommand{\basis}[3]{#1_{#2}\mc\ldots\mc#1_{#3}}

\newcommand{\eps}[1]{\varepsilon}

\newcommand{\norm}[1]{\|#1\|}
\newcommand{\Norm}[1]{\left\|#1\right\|}

\newcommand{\round}[1]{\left\lceil#1\right\rfloor}
\newcommand{\ceil}[1]{\left\lceil#1\right\rceil}

\newcommand{\softO}{\widetilde{O}}
\newcommand*\keywords[1]{%
	\begingroup\small\textbf{Keywords:}\quad #1%
	\par\vspace*{1mm}\endgroup}
\newcommand{\topic}[1]{
	\ifdefined\textemph
		\smallskip\noindent\textemph{#1}~
	\else
		\smallskip\noindent\textbf{#1}~
	\fi		
}
\newcommand{\Dep}{\mathsf{Dep}}

\newcommand{\mc}{\mbox{,\ }}

\newcommand{\rank}{\mathrm{rank}}

\newcommand{\HNF}{\mathrm{HNF}}

\newcommand{\rational}{\mathbb{Q}}
\newcommand{\integer}{\mathbb{Z}}

\newcommand{\T}{\mathrm{T}}

\renewcommand{\to}[3]{#1\mc#2\mc\ldots\mc #3}

\newcommand{\sample}{\leftarrow}

\newcommand{\MDep}{\mathsf{MDep}}

\title{\textbf{On the probability of generating a primitive matrix}\thanks{Jingwei \uppercase{Chen}.
		Chongqing Key Laboratory of Automated Reasoning and Cognition, Chongqing Institute of
		Green and Intelligent Technology, Chinese Academy of Sciences, Chongqing $400714$, China; Chongqing College, University of Chinese Academy of Sciences, Chongqing $400714$, China; Key Laboratory of Advanced Manufacturing Technology of Ministry of Education, Guizhou
		Universityy, Guiyang $550025$, China.  Email: chenjingwei@cigit.ac.cn \\  
		Yong \uppercase{Feng}, Wenyuan \uppercase{Wu}. 
		Chongqing Key Laboratory of Automated Reasoning and Cognition, Chongqing Institute of
		Green and Intelligent Technology, Chinese Academy of Sciences, Chongqing $400714$, China; Chongqing College, University of Chinese Academy of Sciences, Chongqing $400714$, China.  Email: \{yongfeng, wuwenyuan\}@cigit.ac.cn\\
		Yang  \uppercase{Liu} (Corresponding author). 
		Information Science and Engineering, Chongqing Jiaotong University, Chongqing 400074,
		China.  Email: liuyang13@cqjtu.edu.cn\\
	{This research was supported by National Key Research and Development Project (2020YFA0712303), NSFC (61903053), Youth Innovation Promotion Association of CAS, Guizhou Science and Technology Program [2020]4Y056 and Chongqing Science and Technology Program (cstc2021jcyj-msxmX0821, cstc2020yszx-jcyjX0005,  cstc2021yszx-jcyjX0004, 2022YSZX-JCX0011CSTB).}}}
\author{
	Jingwei \uppercase{Chen}\and Yong \uppercase{Feng}\and Yang \uppercase{Liu}\and Wenyuan \uppercase{Wu}
}
\date{\today}

\begin{document}

\maketitle

\begin{abstract}
Given a $k\times n$ integer \textit{primitive}  matrix $\bm{A}$ (i.e., a matrix can be extended to an $n\times n$ unimodular matrix over the integers) with the maximal absolute value of entries $\norm{\bm{A}}$  bounded by  {an integer} $\lambda$ from above, we study the probability that the $m\times n$ matrix extended from $\bm{A}$ by appending other $m-k$ row vectors of dimension $n$ with entries chosen randomly and independently from the uniform distribution over $\{\to{0}{1}{\lambda-1}\}$ is still primitive. We present a complete and rigorous proof of a lower bound on the probability, which is at least a constant for fixed $m$ in the range $[k+1, n-4]$. As an application, we prove that there exists a fast Las Vegas algorithm that completes a $k\times n$ primitive matrix $\bm{A}$ to an  $n\times n$ unimodular matrix within expected  $\softO(n^{\omega}\log \norm{\bm{A}})$ bit operations, where  $\softO$ is big-$O$ but without log factors, $\omega$ is the exponent on the arithmetic operations of matrix multiplication.
\end{abstract}



\section{Introduction}\label{sec:Intro}

A vector $\bm{x}\in\integer^n$ is called \textit{primitive} if $\bm{x}=d\bm{y}$ for $\bm{y}\in\integer^n$ and $d\in\integer$ implies $d=\pm 1$. More generally, a matrix $\bm{A}\in\integer^{k\times n}$ with $k\le n$ is called \textit{primitive} if $\bm{x}=\bm{yA}\in\integer^n$ for $\bm{y}\in\rational^k$ implies $\bm{y}\in\integer^k$; in this case we also say the $k$ rows of $\bm{A}$ are \textit{primitive} in $\integer^n$. In particular, an $n\times n$ primitive matrix over $\integer$ is also called \emph{unimodular}, i.e., an integer square matrix with determinant $\pm 1$. It can be proved that a $k\times n$ primitive matrix can always be extended to an $n\times n$ unimodular matrix over $\integer$; see, e.g., \cite{MazeRosenthalWagner2011}.

Given a  primitive matrix $\bm{A}\in\integer^{k\times n}$ with $\norm{\bm{A}}:=\max_{i,j}\abs{a_{i,j}}$  {bounded by an integer $\lambda$ from above}, our focus in this paper will be on the probability of that the $m\times n$ matrix extended from $\bm{A}$ by appending other $m-k$ vectors of dimension $n$ with entries chosen randomly and independently from the uniform distribution over $\Lambda:=\integer\cap[0\mc\lambda)$ is still primitive.  

\subsection{Main results}
In particular, we prove the following

\begin{thm}\label{thm:main}
	Given a primitive matrix $\bm{A}\in\integer^{k\times n}$ with $\norm{\bm{A}}$  {bounded by an integer $\lambda$ from above} and an integer $s$ with $0\le s\le n-k-2$, let $\bm{B}\in\integer^{(n-s-1)\times n}$ be a matrix with first $k$ rows  copied from $\bm{A}$ and entries of  the other rows chosen  randomly and independently from the uniform distribution over $\Lambda$. Then the probability of the event that $\bm{B}$ is	primitive is at least 
	\begin{equation}\label{eq:main}
		1-4 \left (\frac{2}{3}\right )^{s+1}\left (1-\left (\frac{2}{3}\right )^{n-k-s-1}\right ) - \dfrac{2(n-s)^2}{\lambda^{s+2}}\left (1-\frac{1}{\lambda^{n-k-s-1}}\right ).
	\end{equation}
\end{thm}

Note that Theorem \ref{thm:main} holds for $k=0$ as well, which is the case of directly choosing $n-s-1$ vectors with entries from the uniform distribution over $\Lambda$. Roughly speaking, for this case,  it was shown by Maze \textit{et al.} \cite{MazeRosenthalWagner2011} that when $\lambda\rightarrow\infty$, the limit probability of that an  $(n-s-1)\times n$ integer matrix with entries random chosen  from the uniform distribution over $\Lambda$  is  primitive is
\begin{equation}\label{eq:natural_density}
	\prod_{j=s+2}^{n}\dfrac{1}{\zeta(j)}, 
\end{equation}   
where $\zeta(\cdot)$ is the Riemann's zeta function. In this sense, Theorem \ref{thm:main} gives an effective lower bound on the probability for finite $\lambda$, and hence will be useful in practice, especially in computer science. 

From Eq. \eqref{eq:main}, the parameter $k$ plays a very limited role for the result. In fact, one may easily obtain a simpler but  worse bound: 
\begin{align}\label{eq:simpler_bnd}
	1-4\left (\frac{2}{3}\right )^{s+1} - \frac{2(n-s)^2}{\lambda^{s+2}}. 
\end{align}
Note that this bound is independent of the parameter $k$.  {For example, if $s=3$ is fixed and $\lambda \ge 3(n-3)^{2/5}$, then the bound \eqref{eq:simpler_bnd} implies that the resulting $(n-4)\times n$ matrix will be primitive with a probability at least $0.2$}.

Moreover, if $\lambda$ is large enough (with respect to $n$), this bound can be further simplified as 
\begin{align}
	\label{eq:oversimplified}
	1-(4+\delta) \left (\frac{2}{3}\right )^{s+1}	
\end{align}
for some $0<\delta<1$. Surprisingly, this oversimplified bound only depends on $s$. 

For given $k\mc n$ and $\lambda$, one can decide the smallest integer $s\in [0,n-k-2]$ such that the lower bound given in Eq. \eqref{eq:main} is a usable bound, i.e., between $0$ and $1$. For instance, $s$ should be at least $3$ for $4 \left (\frac{2}{3}\right )^{s+1}<1$. 

We remark that when $s=n-k-2$ the probability bound given in Theorem \ref{thm:main} matches the empirical probability well according to our experiments in Section \ref{sec:exp}. 

{In addition, one may not further expect a constant probability for the case of $s=-1$ (that corresponds to the resulting matrix is an $n\times n$ unimodular matrix),  since the  natural density of  random $n\times n$ unimodular matrices is $0$; see \cite[Lemma 5]{MazeRosenthalWagner2011}.} 

\subsection{Implications} 
As an application of Theorem \ref{thm:main}, we present a fast Las Vegas algorithm (Algorithm \ref{algo:umc}) that efficiently completes a primitive matrix $\bm{A}\in\integer^{k\times n}$ to an $n\times n$ unimodular matrix $ \bm{U} $ such that $ \norm{\bm{U}} \le n^{O(1)}\norm{\bm{A}} $. More specifically,  we prove the following

\begin{thm}\label{thm:umc}
	Given a primitive matrix $\bm{A}\in\integer^{k\times n}$, there exists a Las Vegas algorithm that completes $\bm{A}$ to an $n\times n$ unimodular matrix $ \bm{U} $ such that $ \norm{\bm{U}} \le n^8\norm{\bm{A}} $ in an expected number of $O(n^{\omega+\varepsilon}\log^{1+\varepsilon}\norm{\bm{A}})$ bit operations.
\end{thm}

\subsection{Techniques} 
The essential ingredient of our proof for Theorem \ref{thm:main} is adapted from \cite[Section 6]{EberlyGiesbrechtVillard2000}, which was used to analyze the expected number of nontrivial invariant factors of a random integer matrix. The main idea is to give an upper bound on the probability that the resulting $(n-s-1)\times n$ matrix is not primitive. Based on the following lemma (whose proof can be obtained from, e.g., \cite[Fact 3]{ChenStorjohann2005:latcompr}), the bound can be analyzed by localizing at each prime $p$.
\begin{lem}
	A $k\times n$ integer matrix $\bm{A}$ is not primitive if and only if there exists at least one prime number $p$ such that the resulting matrix is not full rank over the finite field $\integer_p:=\integer/p\integer$.
\end{lem}

In addition, the algorithm for unimodular matrix completion is based on the determinant reduction technique, which was originally introduced by Storjohann in \cite[Section 15]{Storjohann2003} for computing the determinant of  a polynomial matrix, with a worked example for integer matrix followed. More details about determinant reduction for integer matrices  and an iterated usage of  this technique are given in \cite[Section 13.2]{Storjohann2005}. We give a full description of the algorithm for the integer matrix case and present a detailed analysis in Section \ref{sec:umc}.

\subsection{Related work}   
Primitive and unimodular matrices have many applications in different areas. For example, unimodular matrices can be applied to signal compression \cite{PhoongLin2002}; the lattice reduction algorithms \cite{LenstraLenstraLovasz1982, Schnorr1987} essentially produce a series of unimodular matrices (linear transformations) to improve the basis quality of the given lattice. 

In particular, generating a primitive or unimodular matrix with given rows or columns happens quite often in practice. For instance, one may need to generate unimodular matrices with at least one column of all ones in linear programming for simplex pivoting. In the literature, there exist many results on what conditions should be satisfied for that a partial integral matrix can be completed to a unimodular matrix. In 1956 Reiner \cite{Reiner1956} proved that a row vector can be completed to a unimodular matrix if and only if it is primitive. Zhan \cite{Zhan2006} proved that if $n$ entries of an $n\times n$ partial integral matrix are prescribed and these $n$ entries do not constitute a row or a column, then this matrix can be completed to a unimodular matrix. Fang \cite{Fang2007} improved Zhan's result by proving that if an $n\times n$ partial integral matrix has $2n-3$ prescribed entries and any $n$ entries of these do not constitute a row or a column, then it can be completed to a unimodular matrix. Duffner and Silva \cite{DuffnerSilva2017} gave necessary and sufficient conditions for the existence of unimodular matrices with a prescribed submatrix over a ring that either is Hermite and Dedekind finite or has stable range one. 

For the probability analysis, Maze \textit{et al.} \cite{MazeRosenthalWagner2011} analyzed the natural density of $k\times n$ primitive matrices. Guo \textit{et al.} \cite{GuoHouLiu2016} extended Maze \textit{et al.}'s result to a more general setting, where the natural density of $k\times n$ primitive matrices over all $m\times m$ ($m\ge\max\{k\mc n\}$) integer matrices  was considered. Note that the natural density can be interpreted as a limit probability, where each matrix entry is randomly and independently chosen from a set with an upper bound but the bound tends to infinity. However, a finite version is usually preferable in practice,  e.g., algorithm analysis. Therefore, the result in Theorem \ref{thm:main} will be useful. 

Additionally, a somewhat ``dual'' case is considered in \cite{FonteinWocjan2014}, where the probability of that $m$ integeral vectors with bounded entries generate a same lattice of rank $n$ was studied. In \cite{FonteinWocjan2014}, the ideal choice is $m = n+1$, but a theoretical lower bound on the probability was only proven for $m\ge 2n+1$. Aggarwal and Regev in  \cite{AggarwalRegev2016} and Kirshanova \textit{et al.} in \cite{KirshanovaNguyenStehleWallet2020} considered a closely related problem but the entries are randomly chosen from the discrete gaussian distribution over $\integer^n$, where $m$ is even larger.

For algorithms, Randall \cite{Randall1991} presented an algorithm  for generating random matrices over a finite field with a given determinant, which naturally can be used to generate unimodular matrices over finite fields. Kalaimani \textit{et al.} \cite{KalaimaniBelurSivasubramanian2013} and Zhou and Labahn \cite{ZhouLabahn2014} discussed algorithms for unimodular completion of polynomial matrices. The unimodular matrix completion algorithm discussed in this paper works for integer matrices, which is more efficient than the standard method for this problem (see Remark \ref{rem:umc} and \ref{rem:final}). 

\topic{Roadmap.} We prove Theorem \ref{thm:main} in Section \ref{sec:proof}.  Combining  the determinant reduction technique for integer matrices with Theorem \ref{thm:main}, we give a fast algorithm for the problem of unimodular matrix completion in Section \ref{sec:umc}. In Section \ref{sec:exp}, we present an extensive experimental study on the probability that the resulting $(n-s-1)\times n$ matrix is primitive and discuss some interesting problems for further study.


\section{Proof of Theorem \ref{thm:main}}\label{sec:proof}

Given a primitive matrix $\bm{A}\in\integer^{k\times n}$ with $k<n$ and $\norm{\bm{A}}\le \lambda$, we now consider to extend $\bm{A}$ to an $(n-s-1)\times n$ matrix  by choosing other $n-k-s-1$ vectors with entries randomly and independently chosen from the uniform distribution over $\Lambda = \integer\cap[0\mc\lambda)$, where the integer $s$ satisfies $0\le s\le n-k-2$. Denote by $\bm{a}_i=(a_{i,j})_{1\le j\le n}$ the $i$-th row of $\bm{A}$. Then $\norm{\bm{a}_i}_\infty:=\max_j\{\abs{a_{i,j}}\}\le \lambda$. We always assume that $\lambda\ge 2$ for excluding the case of $\Lambda=\{0\}$. For convenience, we still use $\basis{\bm{a}}{k+1}{n-s-1}\in\integer^{n}$ to denote the random vectors with each entry chosen randomly and independently from the uniform distribution over $\Lambda$, and  denote 
\begin{align*}
	\bm{A}_i = \begin{pmatrix}
		\bm{a}_1\\\bm{a}_2\\\vdots\\\bm{a}_i
	\end{pmatrix} = \begin{pmatrix}
		a_{1,1} & a_{1,2}&\cdots& a_{1,n}\\
		a_{2,1} & a_{2,2}&\cdots& a_{2,n}\\
		\vdots&		\vdots&&\vdots\\
		a_{i,1}&a_{i,2}&\cdots&a_{i,n}
	\end{pmatrix}\mc\quad i=k\mc k+1\mc \ldots\mc n-s-1.
\end{align*} 

To prove Theorem \ref{thm:main}, we firstly need to bound from above the probability $P$ of the event that  the matrix $\bm{A}_{n-s-1}$ is not primitive under the assumption  that the given matrix $\bm{A}$ is primitive. 

\begin{lem}\label{lem:pr_n-s-1}
	Let all notations be as above. Then
	\begin{align*}
		P\le  4 \left (\frac{2}{3}\right )^{s+1}\left (1-\left (\frac{2}{3}\right )^{n-k-s-1}\right ) + \dfrac{2(n-s)^2}{\lambda^{s+2}}\left (1-\frac{1}{\lambda^{n-k-s-1}}\right ).
	\end{align*}
\end{lem}

We prove Lemma \ref{lem:pr_n-s-1} following the approach of Eberly \textit{et al.} \cite[Section 6]{EberlyGiesbrechtVillard2000}, whose original goal was to bound the expected number of invariant factors for random integer matrices. For $k\le i\le n-s-1$, we define the event 	
\begin{enumerate}[noitemsep, topsep=0pt]
	\item[$\bullet$] $\MDep_i$: There exists at least one prime number $p$ such that $\rank(\bm{A}_i)\le i-1$ over $\integer_p$. 
\end{enumerate}

\begin{rem}
	The definition of $\MDep_{i}$ here is different from that in \cite[Section 6]{EberlyGiesbrechtVillard2000}, where $\MDep_{i}$ denotes the event that there exists at least one prime $p$ such that $\rank(\basis{\bm{a}}{1}{i})\le i-2 $ over $ \integer_p $.
\end{rem}

So the assumption that  $\bm{A}$ is primitive is equivalent to that the event $\neg\MDep_{k}$ happens. Under the assumption, we have
\begin{equation}\label{eq:p_n-3}
	\begin{aligned}
		P  &= \Pr[\MDep_{n-s-1}]\le\Pr[\MDep_{k+1}\vee\MDep_{k+2}\vee\cdots\vee\MDep_{n-s-1}]\\
		&=\Pr[\MDep_{k+1}\vee(\MDep_{k+2}\wedge\neg\MDep_{k+1})\vee\cdots\vee (\MDep_{n-s-1}\wedge\neg\MDep_{n-s-2})] \\
		&=\Pr[(\MDep_{k+1}\wedge\neg\MDep_{k})\vee(\MDep_{k+2}\wedge\neg\MDep_{k+1})\vee\cdots\vee(\MDep_{n-s-1}\wedge\neg\MDep_{n-s-2})]\\
		&\le \sum_{i=k+1}^{n-s-1}\Pr[\MDep_{i}\wedge\neg\MDep_{i-1}] \le  \sum_{i=k+1}^{n-s-1}\Pr[\MDep_{i}\,\vert\,\neg\MDep_{i-1}].
	\end{aligned}
\end{equation}
In order to bound $ \Pr[\MDep_{i}\,\vert\,\neg\MDep_{i-1}] $ for $k+1\le i\le n-s-1$, we need to introduce another useful event:
\begin{enumerate}[noitemsep, topsep=0pt]
	\item[$\bullet$] $\Dep_i$:  The rows of $\bm{A}_i$ are linearly dependent over $\rational$, i.e., $\rank(\bm{A}_i)\le i-1$ over $\rational$. 	
\end{enumerate}
Now we have
\begin{equation} 
	\label{eq:infinite-finite}
	\begin{aligned}
		\Pr[\MDep_{i}\,\vert\,\neg\MDep_{i-1}]  
		=\ &  \Pr[\MDep_{i}\wedge(\Dep_i\vee\neg\Dep_i)\,\vert\,\neg\MDep_{i-1}]\\
		=\ &  \Pr[(\MDep_{i}\wedge\Dep_i)\vee(\MDep_{i}\wedge\neg\Dep_i)\,\vert\,\neg\MDep_{i-1}]\\
		\le\ & \Pr[(\MDep_{i}\wedge\Dep_i)\,\vert\,\neg\MDep_{i-1}] + \Pr[(\MDep_{i}\wedge\neg\Dep_i)\,\vert\,\neg\MDep_{i-1}].
	\end{aligned}
\end{equation}
We first bound 
\[
\Pr[(\MDep_{i}\wedge\Dep_i)\,\vert\,\neg\MDep_{i-1}]\le \Pr[\Dep_i\,\vert\,\neg\MDep_{i-1}].
\] 
The latter is the probability of that $\rank(\bm{A}_{i}) = i -1$ over $\rational$ on the condition that $\bm{A}_{i-1}$ is primitive.  From $\bm{A}_{i-1}$ is primitive, it follows that there must exist $i-1$ columns of $\bm{A}_{i-1}$ such that the submatrix consisting of the first $i-1$ rows of these $i-1$ columns has rank $i-1$ over $\rational$. Denote by the set of indices of these columns $C_{i-1}$. Now $\rank(\bm{A}_{i}) = i -1$ over $\rational$ implies that for all $j\notin C_{i-1}$, the entry $a_{i, j}$ must be a linear combination of $a_{\ell,j}$'s for $\ell = 1,\ldots, i-1$ with the same rational coefficients that determined by those $a_{i, j}$'s with $j\in C_{i-1}$. However, each entry of $\bm{a}_i$ are chosen randomly and independently from the uniform distribution over $\Lambda$. Thus, for each  $j\notin C_{i-1}$, the likelihood that $a_{i,j}$ is equal to such a rational linear combination is either $0$ or $\frac{1}{\lambda}$. Therefore
\begin{align}
	\label{eq:infinite}
	\Pr[(\MDep_{i}\wedge\Dep_i)\,\vert\,\neg\MDep_{i-1}] \le \Pr[\Dep_i\,\vert\,\neg\MDep_{i-1}]\le \left(\frac{1}{\lambda}\right)^{n-i+1}.
\end{align}

To bound $\Pr[(\MDep_{i}\wedge\neg\Dep_i)\,\vert\,\neg\MDep_{i-1}]$, let us consider primes $p<\lambda$ and primes $p\ge \lambda$, respectively.
For that, 	we define the following events:
	\begin{enumerate}[noitemsep, topsep=0pt]
		\item[$\bullet$]  $\MDep_i^{(p)}$: For prime $p$, $\rank(\bm{A}_i)\le i-1$ over $\integer_p$. 
		\item[$\bullet$] $\MDep_i^{(p<\lambda)}$: there exists a prime $p<\lambda$ such that $\rank(\bm{A}_i)\le i-1$ over $\integer_p$.
	\end{enumerate}

\subsection[Case 1]{The case of $p< \lambda$} 

We bound $\Pr[(\MDep_{i}^ {(p<\lambda)}\wedge\Dep_i)\,\vert\,\neg\MDep_{i-1}]$  by $\Pr[\MDep_{i}^ {(p<\lambda)}\,\vert\,\neg\MDep_{i-1}]$.
The latter one is the probability that there exists a prime $p<\lambda$ such that $\rank(\bm{A}_i) = i-1$ over $\integer_p$ under the condition that $\bm{A}_{i-1}$ is primitive. Now we  bound the latter one case by case. 

If  $\lambda =2$, then no such prime $ p $ exists. 

If $\lambda=3$, then $ p=2 $. Furthermore, the event  $ (\MDep_{i}^ {(p<3)}\,\vert\,\neg\MDep_{i-1}) $ is that  $\rank(\bm{A}_{i})=i-1$ over $\integer_2$ assuming $\rank(\bm{A}_{i-1})= i-1 $  {over $\integer_p$ for any prime $p$.  It follows from the assumption} that there exist $i-1$ columns  of $\bm{A}_{i}$ such that the submatrix consisting of the first $i-1$ rows of these $ i-1 $ columns has rank $i-1$ over  $\integer_2$. Denote by the set of indices of these columns $C_{i-1}$. Now, $\rank(\bm{A}_{i})=i-1$ for $ p=2 $ means that for each $j\notin C_{i-1}$, the entry $a_{i, j}$ must be a linear combination of $a_{\ell,j}'s$ over $\integer_2$ for $\ell=1,\ldots, i-1$. However, when $\lambda =3$ and $p=2$, we have
\begin{align*}
	\Pr_{x\sample\Lambda}\left [x\equiv 0\mod p\right ] =\frac{2}{3} 
\end{align*}
and
\begin{align*}
	\Pr_{x\sample\Lambda}\left [x\equiv 1\mod p\right ] = \frac{1}{3},
\end{align*}
so for $\lambda =3$ it follows that
\begin{align*}
	\Pr\left [\MDep_{i}^ {(p<3)}\,\vert\,\neg\MDep_{i-1}\right ] \le  \left (\frac{2}{3}\right )^{n-i+1}.
\end{align*}

If $\lambda=4$, then $p=2$ or $p=3$. Similarly, we have 
\begin{align*}
	\Pr\left [\MDep_{i}^ {(p<4)}\,\vert\,\neg\MDep_{i-1}\right ] \le 2\left (\frac{1}{2}\right )^{n-(i-1)}\le \left (\frac{2}{3}\right )^{n-i+1}
\end{align*}
for $i< n-1$.

If $\lambda=5$, then $p=2$ or $p=3$, and further we obtain
\begin{align*}
	\Pr\left [\MDep_{i}^ {(p<5)}\,\vert\,\neg\MDep_{i-1}\right ]  \le \left (\frac{3}{5}\right )^{n-(i-1)} + \left (\frac{2}{5}\right )^{n-(i-1)}\le \left (\frac{2}{3}\right )^{n-i+1}
\end{align*}
for $i< n-1$.

If $\lambda=6$, then $p=2$ or $p=3$ or $p=5$ and we have
\begin{align*}
	\Pr\left [\MDep_{i}^ {(p<6)}\,\vert\,\neg\MDep_{i-1}\right ] \le \left (\frac{1}{2}\right )^{n-i+1} + 2\left (\frac{1}{3}\right )^{n-(i-1)} \le \left (\frac{2}{3}\right )^{n-i+1}
\end{align*}
for $i< n-1$.

If $\lambda=7$, then $p=2$ or $p=3$ or $p=5$. It follows that
\begin{align*}
	\Pr\left [\MDep_{i}^ {(p<7)}\,\vert\,\neg\MDep_{i-1}\right ] \le \left (\frac{4}{7}\right )^{n-i+1} + \left (\frac{3}{7}\right )^{n-(i-1)} + \left (\frac{2}{7}\right )^{n-(i-1)} \le \left (\frac{2}{3}\right )^{n-i+1}
\end{align*}
for $i< n-1$.

If $\lambda\ge 8$, then $p=2$ or $p=3$ or $p=5$ or $p=7$, etc. Then for $i< n-1$, we have
\begin{align*}
	\Pr\left [\MDep_{i}^ {(p<\lambda)}\,\vert\,\neg\MDep_{i-1}\right ] &\le \left (\frac{1}{2}\right )^{n-i+1} + \left (\frac{3}{8}\right )^{n-(i-1)} + \left (\frac{1}{4}\right )^{n-(i-1)} +\sum_{p\ge 7}\left (\frac{2}{p-1}\right )^{n-i+1}\\
	&\le  \left (\frac{2}{3}\right )^{n-i+1} +  \left (\frac{1}{3}\right )^{n-i+1}\cdot\sum_{p\ge 7}\frac{4}{(p-1)^2}\\
	&\le \left (\frac{2}{3}\right )^{n-i+1} +  4\cdot \left (\frac{1}{3}\right )^{n-i+1}\cdot \sum_{j=6}^{\infty}\frac{1}{j^2}\\
	&\le \left (\frac{2}{3}\right )^{n-i+1} +  4\cdot \left (\frac{1}{3}\right )^{n-i+1}\cdot\left (\zeta(2)-1-\frac{1}{4} - \frac{1}{9}-\frac{1}{16} - \frac{1}{25}\right )\\
	&\le  \left (\frac{2}{3}\right )^{n-i+1} +  \frac{3}{4}\left (\frac{1}{3}\right )^{n-i+1},
\end{align*}
where $\zeta(\cdot)$ is Riemman's zeta function and the fact that $ \ceil{\lambda/p}/\lambda $ does not increase with respect to $\lambda$ and that $ \ceil{\lambda/p}/\lambda \le 2/(p-1)$ for $p<\lambda$ were used. Therefore, for the case of $p< \lambda$ we proved the following 
\begin{prop}\label{prop:plelambda}
	Let $\lambda\ge 2$ be an integer and $k+1\le i\le n-3$, and suppose that the event $\neg\MDep_{i-1}$ happens. The probability that there exists any prime $p<\lambda$ such that $\rank(\bm{A}_i) \le i-1$ over $\integer_p$ 
	is at most  $\left (\frac{2}{3}\right )^{n-i+1} +  \frac{3}{4}\left (\frac{1}{3}\right )^{n-i+1}$. In particular,
	\begin{align}\label{eq:p_lt_lambda}
		\Pr\left [(\MDep_{i}^{(p<\lambda)}\wedge\neg\Dep_i)\,\vert\,\neg\MDep_{i-1}\right ] \le \left (\frac{2}{3}\right )^{n-i+1} +  \frac{3}{4}\left (\frac{1}{3}\right )^{n-i+1}.
	\end{align}
\end{prop}

\subsection[Case 2]{The case of $p\ge \lambda$} 
Now our goal is to bound 
\begin{align}\label{eq:pgtlambda}
	\begin{aligned}
		\Pr\left [\left (\MDep_{i}^{(p\ge\lambda)}\wedge\neg\Dep_i\right )\,\bigg\vert\,\neg\MDep_{i-1}\right ]  &= \Pr\left [\left (\bigvee_{p\ge\lambda}\MDep_{i}^{(p)}\right )\wedge\neg\Dep_i\,\bigg\vert\,\neg\MDep_{i-1}\right ]\\
		&=\Pr\left [\bigvee_{p\ge\lambda}\left (\MDep_{i}^{(p)}\wedge\neg\Dep_i\right )\,\bigg\vert\,\neg\MDep_{i-1}\right ]\\
		&\le \sum_{p\ge\lambda} \Pr\left [\left (\MDep_{i}^{(p)}\wedge\neg\Dep_i\right )\,\bigg\vert\,\neg\MDep_{i-1}\right ]
	\end{aligned}
\end{align}
where $p$ ranges all primes at least $\lambda$.

If $p\ge\lambda$ is a fixed prime, then the probability that $\rank(\bm{A}_i) \le i-1$ over $\integer_p$  {under the condition  $\rank(\bm{A}_{i-1})=i-1$ over $\rational$} is at most $\left (\frac{1}{\lambda}\right )^{n-i+1}$, since the probability that a value chosen  randomly and independently from the uniform distribution on $\Lambda = \integer\cap[0,\lambda)$ equals a given value in $\integer_p$ is either $0$ or $\frac{1}{\lambda}$, and hence at most $\frac{1}{\lambda}$. So we have 
\begin{align}\label{eq:mdepp}
	\Pr\left [(\MDep_{i}^{(p)}\wedge\neg\Dep_i)\,\vert\,\neg\MDep_{i-1}\right ] \le
	\Pr\left [\MDep_{i}^{(p)}\,\vert\,\neg\MDep_{i-1}\right ] \le \left (\frac{1}{\lambda}\right )^{n-i+1}.
\end{align}
Furthermore, we claim that for all primes $p> (i\cdot\lambda)^i$
\begin{align}\label{eq:finiteness}
	\Pr\left [\left (\MDep_{i}^{(p)}\wedge\neg\Dep_i\right )\,\bigg\vert\,\neg\MDep_{i-1}\right ] = 0.
\end{align}

In fact, if there exists a prime $p>(i\cdot\lambda)^i$ that makes the event $(\MDep_{i}^{(p)}\wedge\neg\Dep_i)$ happen assuming that $\bm{A}_{i-1}$ is primitive, then $p$ must divide all $i\times i$ minors of $\bm{A}_i$. However, although $a_{i,j}$'s $(j=1,2,\ldots, n)$ are  chosen randomly and independently from the uniform distribution on $\Lambda$, once $\neg\Dep_i$ happens, there exists at least one nonsingular $i\times i$ submatrix in $\bm{A}_i$ over $\rational$. The absolute value of the determinant of such a submatrix is  at most $i!\cdot\lambda^{i} \le (i\cdot\lambda)^i$. Combining $p>(i\cdot\lambda)^i$, $p$ divides the determinant, and the absolute value of the determinant is  bounded by $(i\cdot\lambda)^i$, we have the determiant must be zero, which contradicts with non-singularity of the submatrix.  Also, note that  the bound $(i\cdot\lambda)^i$ is independent of the choice of  $a_{i,j}$'s when such  $a_{i,j}$'s make $\neg\Dep_i$ happen.  (Instead, when such  $a_{i,j}$'s do not make $\neg\Dep_i$  happen, the probability $\Pr[(\MDep_{i}\wedge\Dep_i)\,\vert\,\neg\MDep_{i-1}]$ is already discussed in \eqref{eq:infinite}.)

{Now, the number of possible primes $p\ge\lambda$ that divides an integer $d\le (i\cdot\lambda)^i$ is at most $\log_{\lambda} (i\cdot\lambda)^i\le i(1+\log_{\lambda}i)$, where $d$ is in fact the greatest common divisor of all nonsingular $i\times i$ submatrix of $\bm{A}_i$ over $\rational$. Although the integer $d$ may vary as  $a_{i,j}$'s are chosen randomly and independently from the uniform distribution on $\Lambda$, it is always bounded by $(i\cdot\lambda)^i$. Thus, the number of possible primes that we need to consider is at most $i(1+\log_{\lambda}i)$.  Thus, combining Eq. (\ref{eq:pgtlambda}-\ref{eq:finiteness}), we have}

\begin{align}\label{eq:p_g_t_lambda}
	\begin{aligned}
		\Pr\left [(\MDep_{i}^{(p\ge\lambda)}\wedge\neg\Dep_i)\,\vert\,\neg\MDep_{i-1}\right ]		&\le \sum_{p\ge\lambda} \Pr\left [\left (\MDep_{i}^{(p)}\wedge\neg\Dep_i\right )\,\bigg\vert\,\neg\MDep_{i-1}\right ] \le 	\left (i(1+\log_{\lambda}i)\right )\cdot\left (\frac{1}{\lambda}\right )^{n-i+1}.
	\end{aligned}
\end{align}
Therefore, it follows from  Eq. \eqref{eq:p_lt_lambda} and \eqref{eq:p_g_t_lambda} that 

\begin{align}\label{eq:plesslambda}
	\begin{aligned}
		&\Pr\left [(\MDep_{i}\wedge\neg\Dep_i)\,\vert\,\neg\MDep_{i-1}\right ] \\=&\Pr\left [(\MDep_{i}^{(p<\lambda)}\wedge\neg\Dep_i)\,\vert\,\neg\MDep_{i-1}\right ] + \Pr\left [(\MDep_{i}^{(p\ge\lambda)}\wedge\neg\Dep_i)\,\vert\,\neg\MDep_{i-1}\right ]\\
		\le & \left (\frac{2}{3}\right )^{n-i+1} +  \frac{3}{4}\left (\frac{1}{3}\right )^{n-i+1}+\left (i(1+\log_{\lambda}i)\right )\cdot\left (\frac{1}{\lambda}\right )^{n-i+1}.
	\end{aligned}
\end{align}

\paragraph{Proof of Lemma \ref{lem:pr_n-s-1}} It follows from  Eq. \eqref{eq:infinite-finite}, \eqref{eq:infinite}, and \eqref{eq:plesslambda} that for integer $\lambda\ge 2$,  
\begin{align*}
	\Pr[\MDep_{i}\vert\neg\MDep_{i-1}]\le\ & \Pr[(\MDep_{i}\wedge\Dep_i)\,\vert\,\neg\MDep_{i-1}] + \Pr[(\MDep_{i}\wedge\neg\Dep_i)\,\vert\,\neg\MDep_{i-1}]\\
	\le\ &\left (\frac{1}{\lambda}\right )^{n-i+1} + \left (\frac{2}{3}\right )^{n-i+1} +  \frac{3}{4}\left (\frac{1}{3}\right )^{n-i+1}+\left (i(1+\log_{\lambda}i) \right )\cdot\left (\frac{1}{\lambda}\right )^{n-i+1}\\
	\le\ & \left (\frac{2}{3}\right )^{n-i+1} +  \frac{3}{4}\left (\frac{1}{3}\right )^{n-i+1}+\left (i(1+\log_{\lambda}i)+1 \right )\cdot\left (\frac{1}{\lambda}\right )^{n-i+1}.
\end{align*}
Therefore, under the assumption that $\neg\MDep_{k}$ happens, it follows from Eq. \eqref{eq:p_n-3} that 
\begin{align*}
	P & = \Pr[\MDep_{n-s-1}]  \le   \sum_{i=k+1}^{n-s-1}\Pr[\MDep_{i}\vert\neg\MDep_{i-1}]\\
	&\le  \sum_{i=k+1}^{n-s-1}\left (\left (\frac{2}{3}\right )^{n-i+1} +  \frac{3}{4}\left (\frac{1}{3}\right )^{n-i+1}+\left (i(1+\log_{\lambda}i) + 1 \right )\cdot\left (\frac{1}{\lambda}\right )^{n-i+1}\right )\\
	&\le  2\cdot\sum_{i=k+1}^{n-s-1}\left (\frac{2}{3}\right )^{n-i+1} + (n-s)(1+\log_{\lambda}(n-s-1))\sum_{i=k+1}^{n-s-1}\frac{1}{\lambda^{n-i+1}}\\
	&\le 4 \left (\frac{2}{3}\right )^{s+1}\left (1-\left (\frac{2}{3}\right )^{n-k-s-1}\right ) + \dfrac{(n-s)^2}{(\lambda-1)\lambda^{s+1}}\left (1-\frac{1}{\lambda^{n-k-s-1}}\right ) \\
	&\le 4 \left (\frac{2}{3}\right )^{s+1}\left (1-\left (\frac{2}{3}\right )^{n-k-s-1}\right ) + \dfrac{2(n-s)^2}{\lambda^{s+2}}\left (1-\frac{1}{\lambda^{n-k-s-1}}\right ) .\qed
\end{align*}

Now, Theorem \ref{thm:main} is a direct consequence of Lemma \ref{lem:pr_n-s-1}. 

{\begin{rem}
		According to the above proof, we use localization at each prime $p$ and count the number of possibilities over $\integer_p$. So the set $\Lambda$ can be any set containing $\lambda$ contiguous integers, as in \cite[Section 6]{EberlyGiesbrechtVillard2000}. 
\end{rem}}

\section{Unimodular matrix completion}\label{sec:umc}

When completing a given $k\times n$ primitive matrix to an $n\times n$ unimodular matrix, we first complete the input matrix to an $n\times n$  matrix with $n-k$  vectors whose entries are chosen randomly and uniformly from the uniform distribution over $\Lambda$, and then rectify the last four vectors via the determinant reduction technique \cite[Section 15]{Storjohann2003}. Repeat the above process until the resulting matrix is unimodular.

{Eq. \eqref{eq:simpler_bnd} (a consequence of Theorem \ref{thm:main}) with $s=3$ and $\lambda \ge 3(n-3)^{2/5}$} guarantees  that the completed $(n-4) \times n$ matrix is still primitive with  a probability at least a constant, say $0.2$. Then after a constant number of repeat, it is expected that we will obtain at least one $(n-4)\times 4$ primitive matrix, and the determinant reduction technique applies to produce a unimodular matrix.

{
	\subsection{Hermite normal form}
	
	A matrix $\bm{A}\in\integer^{m\times n}$ of rank $r$ has a (row) \textit{Hermite normal form} (HNF) if there exists a square unimodular matrix $\bm{U}$ such that $\bm{H}=\bm{UA}$ satisfies the following:
	there exist indices $1\le i_1 < i_2 < \cdots < i_r\le n$ such that for $j=1,\ldots, r$, $h_{j, i_j}>0$, $h_{j, k}=0$ if $k<i_j$ and $0\le h_{\ell, i_j} < h_{j, i_j}$ if $\ell<j$; the bottom $m-r$ rows of $\bm{H}$ are zero.
	
	The matrix $\bm{H}$ is unique, denoted by $\HNF(\bm{A})$, although the unimodular tranformations are usually not unique. With the definition of HNF, we can character the primitive matrix as the following:
	\begin{lem}\label{lem:primitive_hnf}
		A $k\times n$ integer matrix $\bm{A}$ with $k<n$ is primitive iff $\HNF(\bm{A}^\T) = \binom{\bm{I}_k}{\bm{0}}$.
	\end{lem}
	One can find a proof of this lemma in, e.g., \cite{MazeRosenthalWagner2011}.
}

\subsection{Determinant reduction}\label{subsec:det_red}
Given a nonsingular $\bm{A}\in\integer^{n\times n}$, the determinant reduction introduced in \cite[Section 15]{Storjohann2003} computes a matrix $\bm{B}\in\integer^{n\times n}$, obtained from $\bm{A}$ by replacing the last column, such that the last diagonal entry in the Hermite normal form of $\bm{B}$ is one. 
The determinant reduction was originally presented for integral polynomial matrix 
in \cite[Section 15]{Storjohann2003}, with a worked example for integer matrix. Here, we give full details and analyses for the integer matrix case.

\begin{algorithm}
	\caption{(Determinant reduction)}
	\label{algo:detRed}
	\begin{algorithmic}[1]
		\REQUIRE A nonsingular integer matrix $\bm{A}\in\integer^{n\times n}$. 
		\ENSURE A matrix $\bm{B}\in \integer^{n\times n}$, with $\bm{B}$ equal to $\bm{A}$ except for possibly the last column, $\norm{\bm{B}}\le n^2\norm{\bm{A}}$, and the last diagonal entry of $\HNF(\bm{B})$ equal to one.
		\STATE\label{algostep:compute_u} Set $\bm{C}_{n-1}$ to be the matrix consisting of the first $n-1$ columns of $\bm{A}$. Compute a primitive vector $\bm{u}\in\integer^n$ such that  {$\bm{u}\bm{C}_{n-1}=\bm{0}$}. 
		\STATE\label{algostep:gcd} Call an extended gcd algorithm to compute $\bm{b}\in\integer^n$ such that  { $\bm{u}\bm{b}^\T = 1$.}
		\STATE\label{algostep:sr_ini} Set $\overline{\bm{A}}$ to be the $(n-1)\times (n-1)$ principal submatrix of $\bm{A}$, and $\overline{\bm{b}}$ the vector consisting of the first $n-1$ entries of $\bm{b}$. 
		\STATE\label{algostep:compute_q} Compute  {$\overline{\bm{q}}^\T :=\round{\overline{\bm{A}}^{-1}\overline{\bm{b}}^\T}$}. //For a vector $\bm{v}$, $\round{\bm{v}}$ means each entry rounded.
		\STATE\label{algostep:sr_end} Set $\bm{q}:=(\overline{\bm{q}}\mc 0)$ and set $\bm{B}$ to be $\bm{A}$ except replacing the last column by  {$\bm{b}^\T-\bm{A}\bm{q}^\T$}.
		\RETURN $\bm{B}$.
	\end{algorithmic}
\end{algorithm}

{
	\begin{prop}\label{prop:det_red}
		Given a nonsingular $n\times n$ integer matrix $\bm{A}$, Algorithm \ref{algo:detRed} correctly computes an $n\times n$ integer matrix $\bm{B}$ within $O(n^{\omega+\varepsilon}\log^{1+\varepsilon}\norm{\bm{A}})$ bit operations, where $\bm{B}$ satisfies the following:  
		\begin{itemize}
			\item $\bm{B}$ equals $\bm{A}$ except for possibly the last column,
			\item the last diagonal entry of $\HNF(\bm{B})$ equals one,
			\item $\norm{\bm{B}} = O(n^2\norm{\bm{A}})$.
		\end{itemize}
	\end{prop}
}

\begin{proof}
	{The singularity of $\bm{A}$} implies that the vector $\bm{u}$ produced in Step \ref{algostep:compute_u} is the unique primitive vector in  {$\{\bm{x}\in\integer^n:  \bm{xC}_{n-1} =\bm{0}\}$}. Let $\bm{U}$ be an arbitrary unimodular matrix such that  {$\bm{H}:=\HNF(\bm{A})=\bm{UA}$}. Then the uniqueness of $\bm{u}$ implies that $\bm{u}$ must be the last row of $\bm{U}$. By construction of $\bm{b}$ in Step \ref{algostep:gcd}, the matrix obtained from $\bm{A}$ by replacing the last column with  {$\bm{b}^\T$} will have $\HNF$ with the last diagonal entry one.  The whole algorithm can be expressed as the following equation
	\begin{align*}
		\bm{B} = \begin{pmatrix}
			\overline{\bm{A}} &\overline{\bm{b}} {^\T}\\
			\overline{\bm{a}} & b_{n}
		\end{pmatrix}\cdot\begin{pmatrix}
			\bm{I}_{n-1}& -\overline{\bm{q}} {^\T}\\
			&1
		\end{pmatrix} =\begin{pmatrix}
			\overline{\bm{A}} &\overline{\bm{b}} {^\T} - \overline{\bm{A}}\overline{\bm{q}} {^\T}\\
			\overline{\bm{a}} & b_{n} - \overline{\bm{a}}\,\overline{\bm{q}} {^\T}
		\end{pmatrix},
	\end{align*}
	{where $\overline{\bm{A}}$ is the $(n-1)\times (n-1)$ principal submatrix of $\bm{A}$, $\overline{\bm{b}}$ is the vector consisting of the first $n-1$ entries of $\bm{b}$, and $\bar{\bm{a}}$ is the vector consisting of the first $n-1$ entries of the last row of $\bm{A}$}. Step \ref{algostep:sr_ini}-\ref{algostep:sr_end} are to reduce the size of $\overline{\bm{b}}$ by the columns of $\overline{\bm{A}}$. It follows from  {$\bm{u}\bm{C}_{n-1}=\bm{0}$}  that all entries of  {$\bm{u}\bm{A}$} are zero except for possibly the last one. However, the last entry of $\bm{q}$ equals zero, so  we have $\bm{u}\bm{Aq}^ {\T}=0$, and hence  
	\begin{align}\label{eq:correctness}
		\bm{u} (\bm{b}^ {\T}-\bm{Aq}^ {\T}) = \bm{u}\bm{b}^ {\T}=1,
	\end{align}
	which  {implies that $\bm{UB}$ is exactly $\bm{H}$ except for the last column replaced by $(*,\ldots, *, 1)^\T$, and hence  the last diagonal of $\HNF(\bm{B})$ must be one.} 
	
	We now consider the size of $\bm{b}^ {\T}-\bm{Aq}^ {\T}$. First, 
	\begin{align}\label{eq:bnd_infty}
		\begin{aligned}
			\Norm{\overline{\bm{b}}^ {\T} - \overline{\bm{A}}\overline{\bm{q}}^ {\T} }_\infty &= \Norm{\overline{\bm{b}}^ {\T} - \overline{\bm{A}}\round{\overline{\bm{A}}^{-1}\overline{\bm{b}}^ {\T}}}_\infty\\
			&=\Norm{\overline{\bm{b}}^ {\T} - \overline{\bm{A}}\left(\overline{\bm{A}}^{-1}\overline{\bm{b}}^ {\T}+\bm{\varepsilon}^ {\T}\right) }_\infty\le 
			\Norm{\overline{\bm{A}}\bm{\varepsilon}^ {\T}}_\infty \le\frac{n-1}{2}\norm{\overline{\bm{A}}},
		\end{aligned}
	\end{align}
	where $\norm{\cdot}_\infty$ is the $\ell_\infty$-norm of a vector and $\norm{\bm{\varepsilon}^ {\T}}_\infty\le\frac{1}{2}$ is used. In fact, Step \ref{algostep:compute_q} is essentially the same as Babai's rounding algorithm \cite{Babai1985}. Denote $\bm{u}=(\overline{\bm{u}}, u_n)$. Without loss of generality, we can assume that 
	\begin{align}\label{eq:assumption_on_u}
		\norm{\bm{u}}_\infty = u_n.
	\end{align} Otherwise there must exist a permutation matrix $\bm{P}$ such that $\bm{u}\bm{P}\bm{P}^{-1}\bm{C}_{n-1}=\bm{0}$, $\bm{u}\bm{P}$ satisfies Eq. \eqref{eq:assumption_on_u}, and  $\bm{P}^{-1}\bm{C}_{n-1}$ still corresponds to the first $n-1$ columns of $\bm{A}$ but with a certain column permutation. From Eq. \eqref{eq:correctness}, we have
	$u_n( b_{n} - \overline{\bm{a}}\,\overline{\bm{q}}^ {\T}) =  1-\overline{\bm{u}}(\overline{\bm{b}}^ {\T} - \overline{\bm{A}}\overline{\bm{q}}^ {\T})$,
	combining Eq. \eqref{eq:assumption_on_u}, which gives
	\begin{align*}
		\Abs{b_{n} - \overline{\bm{a}}\,\overline{\bm{q}}^ {\T}}=\frac{1}{\Abs{u_n}}\cdot\Abs{1 -\overline{\bm{u}}(\overline{\bm{b}}^ {\T} - \overline{\bm{A}}\overline{\bm{q}}^ {\T})}\le 1 + \frac{\norm{\overline{\bm{u}}}_\infty}{\abs{u_n}} \cdot \norm{\overline{\bm{b}}^ {\T} - \overline{\bm{A}}\overline{\bm{q}}^ {\T}}_1 \le 1+ \frac{(n-1)^2}{2}\norm{\overline{\bm{A}}},
	\end{align*}
	where $\norm{\cdot}_1$ is the $\ell_1$-norm of a vector,  {the first inequality follows from that $\Abs{\bm{x}\bm{y}^\T}\le\Norm{\bm{x}}_\infty\cdot\Norm{\bm{y}}_1$ holds for any two vectors $\bm{x}$ and $\bm{y}$ of the same dimension, and the second inequality follows from Eq. \eqref{eq:bnd_infty}, \eqref{eq:assumption_on_u}, and $\norm{\bm{x}}_1\le n\norm{\bm{x}}_\infty$ for all $n$-dimensional vector $\bm{x}$}. Therefore, the resulting matrix $\bm{B}$ satisfies $\norm{\bm{B}}  \le n^2\norm{\bm{A}}$.
	
	The cost of Algorithm \ref{algo:detRed} consists in  nonsingular rational linear system solving (Step \ref{algostep:compute_u} and \ref{algostep:compute_q}) that can be finished by a Las Vegas algorithm  in an expected number of $O(n^{\omega+\varepsilon}\log^{1+\varepsilon}\norm{\bm{A}})$ bit operations \cite{MuldersStorjohann2004} and an extended gcd computation (Step \ref{algostep:gcd}) that can be accomplished within $O(n^{2+\varepsilon}\log^{1+\varepsilon}\norm{\bm{A}})$ bit operations \cite[Section 13.2]{Storjohann2005}. Totally, Algorithm \ref{algo:detRed} costs at most $O(n^{\omega+\varepsilon}\log^{1+\varepsilon}\norm{\bm{A}})$ bit operations. 
\end{proof}

\begin{cor}\label{cor:n-1}
	Given an $(n-1)\times n$ primitive matrix $\bm{A}$, there exists a  Las Vegas algorithm which  completes $\bm{A}$ to an $ n\times n $ unimodular matrix in an expected number of  $O(n^{\omega+\varepsilon}\log^{1+\varepsilon}\norm{\bm{A}})$ bit operations. 
\end{cor}

\begin{proof}
	{Since $\bm{A}$ is primitive, Lemma \ref{lem:primitive_hnf} implies that $\HNF(\bm{A}^\T)$ has the form $\binom{\bm{I}_{n-1}}{\bm{0}}$}. Now, one can use the matrix $\bm{B} = \binom{\bm{A}}{\bm{a}_n}^\T$ as the input for Algorithm \ref{algo:detRed}, where $\bm{a}_n$ can be an arbitrary integer vector such that the resulting square matrix is nonsingular.  {Then, $\HNF(\bm{B})$ has the same form as the identity matrix of order $n$ except for the last column.} Thus, the transpose of the output matrix of Algorithm \ref{algo:detRed} will be a unimodular completion of $\bm{A}$.
\end{proof}

\begin{rem}\label{rem:umc}
	To complete an $(n-1)\times n$ primtive matrix to an $n\times n$ unimodular matrix, a standard method is the following: firstly compute $n$ determinants of all  $(n-1)\times (n-1)$ submatrix of the input matrix and then invoking an extended euclidean algorithm will give the information of the last row. However, this standard method can be finished in an expected number of $O(n^{\omega+1+\varepsilon}\log^{1+\varepsilon} \norm{\bm{B}} )$ bit operations, even using the fast Las Vegas  algorithm in \cite{Storjohann2005} for computing the determinant  of an integer matrix in an expected number of $O(n^{\omega+\varepsilon}\log^{1+\varepsilon} \norm{\bm{B}} )$ bit operations.
\end{rem}

\subsection[Proof of Theorem 2]{Proof of Theorem \ref{thm:umc}}
We prove Theorem \ref{thm:umc} by Algorithm \ref{algo:umc}, which uses the Iterated Determinant Reduction presented in \cite[Section 13.2]{Storjohann2005}. Given a nonsingular matrix $\bm{A}\in\integer^{n\times n}$, the Iterated Determinant Reduction technique computes a matrix $\bm{B}\in\integer^{n\times n}$, obtained from $\bm{A}$ by replacing the last $d$ column, such that the last $d$ diagonal entries in the Hermite normal form of $\bm{B}$ is one. The algorithm consists of $d$ times calling of Algorithm \ref{algo:detRed}, each followed with multiplying by a permutation matrix $\bm{P} = \begin{footnotesize}
	\begin{pmatrix}
		\bm{0} &\bm{I}_{n-1}\\
		1 &0
	\end{pmatrix}
\end{footnotesize}$ from right. Now we  are ready to present our algorithm for unimodular matrix completion (Algorithm \ref{algo:umc}).

\begin{algorithm}
	\caption{(Unimodular matrix completion)}
	\label{algo:umc}
	\begin{algorithmic}[1]
		\REQUIRE A primitive matrix $\bm{A}\in\integer^{k\times n}$  {with $k<n$ and $n\ge 5$}. 
		\ENSURE A unimodular matrix $\bm{U}\in \integer^{n\times n}$, with $\bm{U}$ equal to $\bm{A}$ except for the last $n-k$ rows and  $\norm{\bm{U}}\le n^8\norm{\bm{A}}$.
		\STATE\label{algostep:loop} Repeat 
		\STATE\label{algostep:ini} \qquad  {Set $\lambda := \max\{\norm{\bm{A}}, \lceil 3(n-3)^{2/5}\rceil\}$ and $\bm{B}:=\bm{A}^{\T}$.}
		\STATE\label{algostep:completion} \qquad Complete $\bm{B}$ as an $n\times n$ matrix with entries chosen  randomly and    independently from \\
		\qquad  the uniform distribution over $\{0,1\mc\ldots\mc\lambda-1\}$. 
		\STATE\label{algostep:idr} \qquad Repeat the following four times:
		\STATE\label{algostep:dr} \qquad\qquad Set $\bm{B}$ as the output of Algorithm \ref{algo:detRed} with input $\bm{B}$ and then set $\bm{B}:=\bm{BP}$.
		\STATE\label{algostep:until} Until $\det(\bm{B})=\pm 1$
		\RETURN $\bm{U} := (\bm{B}(\bm{P}^{-1})^{4})^{\T}$.
	\end{algorithmic}
\end{algorithm}

Denote $\bm{B}_i$ by the submatrix of $\bm{B}$ consisting of the first $i$ columns. The goal of the loop in Step \ref{algostep:idr} is to make the last four diagonal entries of the Hermite normal form of   $\bm{B}$ be one.  {First, we assume that $\bm{B}_{n-4}^{\T}$ is primitive at the end of Step \ref{algostep:completion}. Then Lemma \ref{lem:primitive_hnf} implies that $$\HNF(\bm{B}) = \begin{pmatrix}
		\bm{I_{n-4}} & \ast & \ast &  \ast&  \ast\\
		& \ast & \ast & \ast & \ast \\
		&  & \ast & \ast & \ast \\
		&  &   & \ast & \ast \\
		&  &  &   & \ast\\
	\end{pmatrix}.$$ 
	After running Algorithm \ref{algo:detRed} once, the HNF of the resulting $\bm{B}$ has the following form
	$$\HNF(\bm{B}) = \begin{pmatrix}
		\bm{I_{n-4}} & \ast & \ast &  \ast&  \bm{0}\\
		& \ast & \ast & \ast & 0 \\
		&  & \ast & \ast & 0 \\
		&  &   & \ast & 0 \\
		&  &  &   & 1\\
	\end{pmatrix}$$ and
	$$\HNF(\bm{BP}) = \begin{pmatrix}
		\bm{I_{n-3}} & \ast & \ast &  \ast\\
		& \ast & \ast & \ast  \\
		&  & \ast & \ast  \\
		&  &   & \ast  \\
	\end{pmatrix}.$$ 
	Therefore, at the end of execution of Step \ref{algostep:idr}, we have $\HNF(\bm{B}) = \bm{I_n}$ for the resluting $\bm{B}$, which implies that $\bm{B}$ is unimodular. 
}

Thanks to Eq. \eqref{eq:simpler_bnd}  {with $s=3$ and $\lambda \ge 3(n-3)^{2/5}$}, after each execution of Step \ref{algostep:completion}, $\bm{B}_{n-4}^{\T}$ is primitive with probability at least $0.2$. Therefore, after  $c$ (a constant) executions of Step \ref{algostep:completion}, the probability that there is at least one $\bm{B}_{n-4}^{\T}$ is primitive is at least $1 - 0.8^c$. So after a constant number of executions of the loop Step \ref{algostep:loop}--\ref{algostep:until}, it is expected that the output of Algorithm \ref{algo:umc} is a unimodular completion of $\bm{A}$. After each execution of Step \ref{algostep:dr},  it follows  from Proposition \ref{prop:det_red} that the magnitude of the new column is bounded by $O(n^2\norm{\bm{B}})$. Therefore, after four executions of Step \ref{algostep:dr}, the magnitude of the new column is bounded by $O(n^8\norm{\bm{A}})$. So the correctness of Algorithm \ref{algo:umc} is proved.

The main cost part of Algorithm \ref{algo:umc} is a constant number of executions of the loop Step  \ref{algostep:loop}--\ref{algostep:until}. According to Proposition \ref{prop:det_red} and  the algorithm for computing determinant in \cite{Storjohann2005}, each execution of the loop can be finished in an expected number of $O(n^{\omega+\varepsilon}\log^{1+\varepsilon}\norm{\bm{A}})$ bit operations. As a consequence, Theorem \ref{thm:umc} is proved.  \qed

\begin{rem}\label{rem:final}
	Without the help of Theorem \ref{thm:main}, one may use the Iterated Determinant Reduction algorithm $n-k$ times for unimodular completion. This results in an algorithm that completes $\bm{A}$ to an $n\times n$ unimodular matrix $\bm{U}$ with $\norm{\bm{U}}\le n^{2(n-k)}\norm{\bm{A}}$ in an expected number of $O((n-k)n^{\omega+\varepsilon}\log^{1+\varepsilon}\norm{\bm{A}})$ bit operations.
\end{rem}

\section{Experiments and discussions}\label{sec:exp}

In this section, we present an  experimental study  on the empirical probability for the case that is included in Theorem \ref{thm:main} ($s\ge 3$), and also for the case that is not included in Theorem \ref{thm:main}, e.g., the case of $s<3$. In all the following experimental data, each empirical probability (labeled Exp.) are obtained by running $10,000$ random tests on  computer algebra system Maple and counting the success rate, i.e., the proportion of primitive matrices among all resulting matrices. 

\subsection[s larger than 3]{The case of $s\ge 3$}

In Table \ref{tab:k_0_s_3} and \ref{tab:k_0_s_4} we study the for the case of $k=0$, so that we can compare the  empirical probability (column labeled Exp.), the lower bound given in Theorem \ref{thm:main} (column labeled Th. 1), and the limit probability given in Eq. \eqref{eq:natural_density} (column labeled Limit probability). We use three different bounds for $\lambda$, namely, $\lambda = 10^5$, $\lambda = 10^{10}$ and $\lambda = 10^{20}$.

\begin{table}[!htbp]
	{\small
		\begin{center}
			\caption{Average empirical probability vs the probability in Theorem \ref{thm:main} ( $k=0$ and $s=3$)}\vskip 1mm
			\label{tab:k_0_s_3}
			
			{\small
				\begin{tabular*}{11.5cm}
					{cccccccc}
					\toprule
					&\multicolumn{2}{c}{$\lambda=10^5$} &\multicolumn{2}{c}{$\lambda=10^{10}$}&\multicolumn{2}{c}{$\lambda=10^{20}$}& \\
					$n$ 	&Exp. &Th. \ref{thm:main}&Exp. &Th. \ref{thm:main}&Exp. &Th. \ref{thm:main}& Limit probability\\[0.8ex]
					\midrule
					$5$		&$0.9652$&$0.7366$&$0.9662$&$0.7366$&$0.9639$&$0.7366$&$0.9643$\\[0.8ex]
					$10$		&$0.9335$&$0.2792$&$0.9291$&$0.2792$&$0.9306$&$0.2792$&$0.9334$\\ [0.8ex]
					$15	$	&$0.9292$&$0.2190$&$0.9278$&$0.2190$&$0.9312$&$0.2190$&$0.9325$\\ [0.8ex]
					$	20$		&$0.9338$&$0.2110$&$0.9349$&$0.2110$&$0.9345$&$0.2110$&$0.9325$\\[0.8ex]
					\bottomrule
			\end{tabular*}}
	\end{center}}
\end{table}

Both Table \ref{tab:k_0_s_3} and \ref{tab:k_0_s_4} shows that the empirical probability is relatively near to the limit probability, both of which are much better than our theoretical lower bound given in Theorem \ref{thm:main}.  As indicated previously, $s=3$ is the smallest $s$ such that the lower bound given in Theorem \ref{thm:main} is between $0$ and $1$ for these experiments, while $s=n-k-2$ is the largest possible value. 

Comparing  Table \ref{tab:k_0_s_3} with \ref{tab:k_0_s_4} shows that larger $s$ implies that  both higher experimental success rate and better theoretical bound,  however, smaller $s$ implies the resulting matrix is closer to a unimodular matrix. In particular, for the case of $s=n-k-2$, the probability bound given in Theorem \ref{thm:main} matches very well with both empirical probability and the limit probability.

\begin{table}[!htbp]
	\begin{center}
		{\small
			\begin{center}
				\caption{Average empirical probability vs the probability in Theorem \ref{thm:main} ($k=0$ and $s=n-k-2$)}\vskip 1mm
				\label{tab:k_0_s_4}
				
				{\small
					\begin{tabular*}{11.5cm}
						{cccccccc}
						\toprule
						&\multicolumn{2}{c}{$\lambda=10^5$} &\multicolumn{2}{c}{$\lambda=10^{10}$}&\multicolumn{2}{c}{$\lambda=10^{20}$}& \\
						$n$ 	&Exp. &Th. \ref{thm:main}&Exp. &Th. \ref{thm:main}&Exp. &Th. \ref{thm:main}& Limit probability\\[0.8ex]\midrule
						$5$		&$0.9652$&$0.7366$&$0.9662$&$0.7366$&$0.9639$&$0.7366$&$0.9643$\\ [0.8ex]
						$10$		&$0.9995$&$0.9653$&$0.9990$&$0.9653$&$0.9990$&$0.9653$&$0.9990$\\ [0.8ex]
						$15$		&$0.9999$&$0.9954$&$1.0000$&$0.9954$&$1.0000$&$0.9954$&$0.9999$\\ [0.8ex]
						$20$		&$1.0000$&$0.9993$&$1.0000$&$0.9993$&$1.0000$&$0.9993$&$0.9999$\\[0.8ex]
						\bottomrule
				\end{tabular*}}
		\end{center}}
	\end{center}
\end{table}

In addition, Table \ref{tab:k_0_s_3} and \ref{tab:k_0_s_4} also show that for different $\lambda$ with same $n$,    the theoretical bounds in column Th. 1 are almost the same and the data of empirical probability are almost the same as well. This is because that for a large enough $\lambda$, the bound given in  Theorem \ref{thm:main} is almost independent of $\lambda$, as indicated by  the oversimplified bound given in \eqref{eq:oversimplified}. For this reason, we fix $\lambda=10^5$ for all other experiments (Table \ref{tab:k_1_s_3}--\ref{tab:k1}).

Tables \ref{tab:k_1_s_3}--\ref{tab:k_n_s_n} are for the case of $k>0$, for which there does not exist a known limit probability. We generate the initial primitive matrix as follows: We first generate a $k\times n$ matrix, whose entries are chosen randomly and independently from the uniform distribution over $[-\lambda\mc\lambda]\cap\integer$. If the matrix is not primitive, we regenerate a new matrix until it is eventually primitive. For each initial matrix, we complete it with uniformly random entries from $\Lambda$ to an $(n-s-1)\times n$ matrix $10,000$ times and count the success rate. From these experiments, we can observe a similar phenomenon as the case of $k=0$. 

\begin{table}[!htbp]
	\caption{Average empirical probability vs the probability in Theorem \ref{thm:main}  ($k=1$, $s=3$)}
	\label{tab:k_1_s_3}
	\centering
	\begin{tabular}{ccccccc}
		\toprule
		$n$ & $10$&$15$&$20$&$25$&$30$\\[0.8ex]\midrule
		Exp. 	&$0.9321$	&$0.9283$	&$0.9291$	&$0.9324$	&$0.9310$	\\[0.8ex]
		Th. \ref{thm:main}	&$0.3139$	&$0.2235$	&$0.2116$	&$0.2101$	&$0.2099$	
		\\[0.8ex]\bottomrule
	\end{tabular}
\end{table}

\begin{table}[!htbp]
	\caption{Average empirical probability vs the probability in Theorem \ref{thm:main}  ($k=1$, $s=n-k-2$)}
	\label{tab:k_1_s_n}
	\centering
	\begin{tabular}{ccccccc}
		\toprule
		$n$ & $5$&$10$&$15$&$20$&$25$&$30$\\[0.8ex]\midrule
		Exp. 				&$0.8919$	&$0.9969$	&$0.9999$	&$1.0000$	&$1.0000$	&$1.0000$	\\[0.8ex]
		Th. \ref{thm:main}	&$0.6049$	&$0.9479$	&$0.9931$	&$0.9990$	&$0.9998$	&$0.9999$	
		\\[0.8ex]\bottomrule
	\end{tabular}
\end{table}

\begin{table}[!htbp]
	\caption{Average empirical probability vs the probability in Theorem \ref{thm:main} ($k=n/2$, $s=3$)}\label{tab:k_n_s_3}
	\centering
	\begin{tabular}{ccccccc}
		\toprule
		$n$ & $16$&$20$&$24$&$28$&$32$&$36$\\[0.8ex]\midrule
		Exp. 				&$0.9349$	&$0.9338$&$0.9340$		&$0.9312$	&$0.9352$	&$0.9333$	\\[0.8ex]
		Th. \ref{thm:main}	&$0.3659$	&$0.2792$	&$0.2407$	&$0.2235$	&$0.2159$	&$0.2125$	
		\\[0.8ex]\bottomrule
	\end{tabular}
\end{table}

\begin{table}[!htbp]
	\caption{Average empirical probability vs the probability in Theorem \ref{thm:main} ($k=n/2$, $s=n-k-2$)}\label{tab:k_n_s_n}
	\centering
	\begin{tabular}{ccccccc}
		\toprule
		$n$ & $16$&$20$&$24$&$28$&$32$&$36$\\[0.8ex]\midrule
		Exp. 				&$0.9919$	&$0.9980$&$0.9996$		&$0.9999$	&$1.0000$	&$1.0000$	\\[0.8ex]
		Th. \ref{thm:main}	&$0.9219$	&$0.9653$	&$0.9845$	&$0.9931$	&$0.9969$	&$0.9986$	
		\\[0.8ex]\bottomrule
	\end{tabular}
\end{table}

Totally, the probability bound presented in Theorem \ref{thm:main} is tight, especially for the case of large $s$. However, how to improve the theoretical bound for small $s$ is an intriguing problem.


\subsection[s smaller than 3]{The case of $s<3$}

For $s<3$, the lower bound on that the resulting $(n-s-1)\times n$ matrix is primitive given in Theorem \ref{thm:main} will be negative and hence useless. Therefore, it would be very interesting and useful to obtain a lower bound for the case of $s<3$. Here we present some experimental results.

\begin{table}[!htbp] 
	\caption{Average empirical probability for the case of $s=2$ and $\lambda=10^5$}\label{tab:s2}
	\centering
	\begin{tabular}{ccccccc}
		\toprule
		$n$ & $16$&$20$&$24$&$28$&$32$&$36$\\[0.8ex]\midrule
		$k=0$ 			&$0.8599$			&$0.8543$	&$0.8575$	&$0.8604$	&$0.8628$	&$0.8643$	\\[0.8ex]
		$k=1$			&$0.8654$			&$0.8618$	&$0.8580$	&$0.8671$	&$0.8646$	&$0.8620$\\[0.8ex]
		$k=\frac{n}{2}$	&$0.8609$			&$0.8611$	&$0.8651$	&$0.8621$	&$0.8682$	&$0.8662$\\[0.8ex]\bottomrule
	\end{tabular}
\end{table}

\begin{table}[!htbp]
	\caption{Average empirical probability for the case of $s=1$ and $\lambda=10^5$}\label{tab:s1}
	\centering
	\begin{tabular}{ccccccc}
		\toprule
		$n$ & $16$&$20$&$24$&$28$&$32$&$36$\\[0.8ex]\midrule
		$k=0$ 				&$0.7201$			&$0.7177$	&$0.7124$	&$0.7227$	&$0.7125$	&$0.7110$	\\[0.8ex]
		$k=1$				&$0.7103$			&$0.7141$	&$0.7154$	&$0.7129$	&$0.7118$	&$0.7192$\\[0.8ex]
		$k=\frac{n}{2}$		&$0.7168$			&$0.7212$	&$0.7210$	&$0.7226$	&$0.7106$	&$0.7154$\\[0.8ex]\bottomrule
	\end{tabular}
\end{table}

\begin{table}[!htbp]
	\caption{Average empirical probability for the case of $s=0$ and $\lambda=10^5$}\label{tab:s0}
	\centering
	\begin{tabular}{ccccccc}
		\toprule
		$n$ & $16$&$20$&$24$&$28$&$32$&$36$\\[0.8ex]\midrule
		$k=0$ 			&$0.4365$	&$0.4363$	&$0.4353$	&$0.4435$	&$0.4434$	&$0.4377$	\\[0.8ex]
		$k=1$			&$0.4323$	&$0.4337$	&$0.4345$	&$0.4451$	&$0.4336$	&$0.4440$\\[0.8ex]
		$k=\frac{n}{2}$	&$0.4385$	&$0.4371$	&$0.4290$	&$0.4330$	&$0.4427$	&$0.4330$\\[0.8ex]\bottomrule
	\end{tabular}
\end{table}

\begin{table}[!htbp]
	\caption{Average empirical probability for the case of $k=1$ and $\lambda=10^5$ with fixed  $(1\mc 1\mc\ldots\mc 1)\in\integer^n$ as the $k\times n$ matrix to be completed.}\label{tab:k1}
	\centering
	\begin{tabular}{ccccccc}
		\toprule
		$n$ 			& $16$		&$20$		&$24$		&$28$		&$32$		&$36$\\[0.8ex]\midrule
		$s=0$ 			&$0.4330$	&$0.4354$	&$0.4342$	&$0.4381$	&$0.4284$	&$0.4407$\\[0.8ex]
		$s=1$			&$0.7163$	&$0.7120$	&$0.7209$	&$0.7198$	&$0.7218$	&$0.7159$\\[0.8ex]
		$s=2$			&$0.8629$	&$0.8641$	&$0.8649$	&$0.8673$	&$0.8616$	&$0.8568$\\[0.8ex]\bottomrule
	\end{tabular}
\end{table}

All test examples in Table \ref{tab:s2}--\ref{tab:s0} are generated with the same method described in the last subsection, i.e., the input matrices are randomly chosen. In Table \ref{tab:k1}, all test examples are fixed to a $1\times n$ primitive row $(1\mc 1\mc\ldots\mc 1)$, but still with $\lambda=10^5$. All of these tables show that for $0\le s\le 2$, the empirical probability of that the resulting matrix is primitive is relatively high, similar with that of the case $s\ge 3$ shown in the last section. However, an effective theoretical lower bound on the probability for this case is left open.

\paragraph{Acknowledgments.} The authors would like to thank  {two anonymous referees} for helpful comments and suggestions that  makes the presentation of this paper clearer.

\let\doi\undefined\newcommand{\doi}[1]{\url{https://doi.org/#1}}\newcommand{\noopsort}[1]{}

\end{document}